\documentclass[12pt,a4paper]{article}
\usepackage{amssymb}
\usepackage{amsfonts}
\usepackage[onehalfspacing]{setspace}

\newtheorem{definition}{Definition}

\newtheorem{lemma}{Lemma}

\newtheorem{proposition}{Proposition}

\newenvironment{proof}[1][Proof]{\noindent\textbf{#1.} }{\ \rule{0.5em}{0.5em}}
\input{tcilatex}
\begin{document}

\title{Behavioral Causal Inference\thanks{%
Financial support by ISF grant no. 320/21 and the Foerder Institute is
gratefully acknowledged. I thank Nathan Hancart and Heidi Thysen, as well as
seminar participants at EUI, Stanford, Berkeley, UCSB and Caltech, for
helpful comments.}}
\author{Ran Spiegler\thanks{%
Tel Aviv University and University College London}}
\maketitle

\begin{abstract}
When inferring the causal effect of one variable on another from
correlational data, a common practice by professional researchers as well as
lay decision makers is to control for some set of exogenous confounding
variables. Choosing an inappropriate set of control variables can lead to
erroneous causal inferences. This paper presents a model of lay decision
makers who use long-run observational data to learn the causal effect of
their actions on a payoff-relevant outcome. Different types of decision
makers use different sets of control variables. I obtain upper bounds on the
equilibrium welfare loss due to wrong causal inferences, for various
families of data-generating processes. The bounds depend on the structure of
the type space. When types are \textquotedblleft ordered\textquotedblright\
in a certain sense, the equilibrium condition greatly reduces the cost of
wrong causal inference due to poor controls.\bigskip \bigskip \pagebreak
\end{abstract}

\section{Introduction}

Learning the causal effect of one variable on another from observational
data is an important economic activity. Indeed, applied economists do it for
a living. However, even lay decision makers regularly perform this activity
in order to evaluate the consequences of their actions. They obtain data
about observed correlations among variables (via first-hand experience or
through the media) and try to extract lessons from the data concerning the
consequences of their actions. For example, which college degree will
improve their long-run economic prospects? Will wearing surgical masks on
airplanes lower their chances of catching a virus? Will drinking more coffee
cause health problems?

There are two main differences between causal inference from observational
data as practiced by professional researchers and lay decision makers.
First, the researcher employs sophisticated inference methods that are
subjected to stringent scrutiny by other professionals. In contrast, lay
decision makers use intuitive, elementary methods, and they do not face any
pushback when they employ these methods inappropriately. The second
difference is that while the professional researcher is an outside observer,
lay decision makers interact with the economic system in question; the
aggregate behavior that results from their causal inferences can affect the
very correlations from which they draw their inferences. Thus, it is apt to
refer to the kind of causal inference that lay decision makers engage in as
\textquotedblleft behavioral\textquotedblright , in both senses of the word.

This paper is an attempt to model \textquotedblleft behavioral causal
inference\textquotedblright . I study a decision maker (DM) who faces a
binary choice between two actions, denoted $0$ and $1$. The DM's choice is
based on his belief regarding the action's causal effect on a binary
payoff-relevant outcome (which also takes the values $0$ or $1$). Using an
intuitive causal-inference method, the DM extracts this causal belief from
long-run correlational data about actions, outcomes and a collection of
exogenous variables. The data is generated by the behavior of other DMs in
similar situations. In equilibrium, the DMs' behavior is consistent with
best-replying to their causal belief.

The intuitive method of causal inference that the DM in my model employs is
very simple: Measuring the observed correlation between actions and
outcomes, while \textit{controlling} for some set of exogenous variables.
This is a basic and widespread procedure in scientific data analysis, but it
is based on a simple idea that lay people practice to some extent. For
example, when an agent decides whether to wear a surgical mask for
protection against viral infection, it is natural for him to look for
infection statistics about people in his own age group. Likewise, when a
student deciding whether to choose a STEM major tries to evaluate the
labor-market outcomes of STEM and non-STEM graduates, it is natural for him
to focus on people who share his highschool math background. In both cases,
when the agent consults data to estimate the consequences of various
actions, he tries to focus on data points that share his own characteristics
--- if he has access to such fine-grained data. This type of controlling
consists of \textit{conditioning} on the realization of some exogenous
variables.

Another type of controlling involves \textit{adjustment} rather than
conditioning. For example, in the above-mentioned surgical-mask example, the
agent may have access to data about the prevalence of certain genes and
their correlation with viral infection. Even if he does not know his own
relevant genetic background, he can nevertheless adjust his beliefs
according to the available data about the correlation between this variable
and others.

In general, suppose that long-run correlational data is given by some joint
probability distribution $p$ over actions $a$, outcomes $y$, and a
collection of exogenous variables is $x_{1},....,x_{K}$. The DM is able to
control for the variables indexed by $D\subseteq \{1,...,K\}$; he conditions
on a subset $C\subseteq D$, and adjusts for the variables in $D\setminus C$.
The DM's estimated causal effect of $a$ on $y$ is given by the formula%
\begin{equation}
\dsum\limits_{x_{D\setminus C}}p(x_{D\setminus C}\mid x_{C})\left[ p(y=1\mid
a=1,x_{D})-p(y=1\mid a=0,x_{D})\right]  \label{formula}
\end{equation}%
When the set $D$ of control variables differs from the set that a outside
researcher would deem appropriate, the DM's causal inference can be wrong:
he may misread the causal meaning of observed correlations, and consequently
obtain a biased estimate of the causal effect of $a$ on $y$.

Erroneous causal inference due to \textquotedblleft bad (exogenous)
controls\textquotedblright\ may take various forms, which are easy to
illustrate with directed acyclic graphs (DAGs), following Pearl (2009). For
instance, suppose that in reality, $a$ has no causal effect on $y$ and that
every observed correlation between these variables is due to confounding by
the exogenous variable $x$. These objective causal relations are represented
by the DAG $a\leftarrow x\rightarrow y$. Given the observed joint
distribution $p$ over $a,x,y$, the proper measurement of the average causal
effect of $a$ on $y$ is given by the formula%
\[
\sum_{x}p(x)[p(y=1\mid a=1,x)-p(y=1\mid a=0,x)] 
\]%
This formula will correctly yield a zero causal effect. If, however, the DM
fails to control for $x$, he will regard $p(y=1\mid a=1)-p(y=1\mid a=0)$ as
the causal effect of $a$ on $y$ --- in other words, he will mistake
correlation for causation --- and potentially measure an erroneous, non-zero
effect.

Bad controls can also involve \textit{excessive }controlling for exogenous
variables. The following example is taken from Cinelli et al. (2022). The
true causal model is given by the DAG $a\leftarrow x_{1}\rightarrow
x_{2}\leftarrow x_{3}\rightarrow y$. Thus, as in the previous example, the
objective causal effect of $a$ on $y$ is null. However, in this case the
quantity $p(y=1\mid a=1)-p(y=1\mid a=0)$ is the correct formula for the
objective (null) causal effect. In other words, there is no need to control
for any of the $x$ variables. Suppose, however, that the DM adjusts for $%
x_{2}$. Then, his estimated causal effect will be%
\[
\sum_{x_{2}}p(x_{2})[p(y=1\mid a=1,x_{2})-p(y=1\mid a=0,x_{2})] 
\]%
In this case, the variable $x_{2}$ is a bad control, and the DM's estimate
can end up being non-null.

This paper poses the following question: What are the limits to the DM's
errors of causal inference due to bad controls, when the data-generating
process $p$ has to be consistent with \textit{equilibrium behavior} ---
i.e., when the DM's choice of actions given his information maximizes his
subjective expected payoff with respect to the belief he extracts from $p$
using his causal-inference procedure?

I study this question with a simple model, in which a DM chooses an action $%
a\in \{0,1\}$ after a collection of exogenous variables $t,x_{1},...,x_{K}$
is realized, where $t\in \{0,1\}$ is the DM's preference type. The DM's vNM
utility function is $u(a,t,y)=y-c\cdot \mathbf{1}[a\neq t]$. Thus, the DM
will only choose $a\neq t$ if he thinks that $a$ has a causal effect on $y$.
In the baseline model, however, I assume that the objective causal effect of 
$a$ on $y$ is null: $y$ is determined only by the exogenous variables
according to some conditional probability distribution (I relax this
assumption in Section 5).

The DM's control variables are given by a \textquotedblleft data
type\textquotedblright\ (drawn randomly from some given set $N$,
independently of his preference type), which is defined by a distinct pair
of subsets $(D,C)$, where: $C\subseteq D\subseteq \{1,...,K\}$; $C$
represents the variables the type conditions on, and $D\setminus C$
represents the variables he only adjusts for, leading to an estimated causal
effect of $a$ on $y$ (given $x$) as described by (\ref{formula}). The
formula is evaluated according to a joint distribution over all variables.
The DM observes the realization of $t$, but he has no long-run data about $t$
and therefore does not use it for causal estimates. In equilibrium, the
distribution of $a$ conditional on the exogenous variables is consistent
with each DM type best-replying to his causal belief. (Section 6 explains
how this concept can be recast in earlier frameworks of equilibrium modeling
with non-rational expectations, due to Jehiel (2005), Spiegler (2016,2020)
and Esponda and Pouzo (2016).)

The basic insight of this paper is that this equilibrium condition can
restrict the magnitude of the DM's welfare loss due to errors of causal
inference. These errors consist of misreading the causal component of
observed correlational patterns. When agents act on these errors, they
change these very patterns, and hence the causal effects they deduce from
them.\bigskip

\noindent \textit{Example 1.1}

\noindent The above pair of examples of \textquotedblleft bad
controls\textquotedblright\ offer an extreme illustration of this insight.
Suppose that $t=0$ with certainty --- i.e., there are no preference shocks.
In the first example, the single exogenous variable $x$ which causes $y$ is
also the sole direct cause of $a$. For the latter causal relation to be
non-null, it must be the case that some DM data types condition their action
on $x$. However, since these types correctly control for $x$, they also
correctly measure the null objective causal effect of $a$ on $y$. Since $t=0$
for sure, these types will play $a=0$ with certainty. By definition, the
same lack of variation of $a$ with $x$ extends to the types who cannot
condition their action on $x$. It follows that no DM type will vary his
action with $x$, which destroys the confounding effect of $x$, and therefore
any causal error due to failure to control for $x$. This means that in
equilibrium, the DM will not incur any welfare loss due to poor causal
inference.

The same reasoning applies to the second example, which involves three
exogenous $x$ variables. If a DM data type conditions on $x_{1}$, his causal
inference is sound and therefore his action is constant (since there are no
preference shocks); whereas if he does not condition on $x_{1}$, his
behavior is independent of $x_{1}$ by definition. Thus, no DM type varies
his behavior with $x_{1}$, such that the link $a\leftarrow x_{1}$ that makes 
$x_{2}$ a \textquotedblleft bad control\textquotedblright\ is effectively
severed. $\square $\bigskip 

The main results in this paper --- presented in Sections 3 and 4 --- explore
the generality of this observation. I examine various families of joint
distributions over $t,x_{1},...,x_{K},y$, and characterize the upper bound
on the DM's equilibrium welfare loss relative to the rational-expectations
strategy $a\equiv t$. Since $a$ has no causal effect on $y$ in the baseline
model, the welfare loss is simply $c$ times the probability of playing $%
a\neq t$.

It turns out that a simple binary relation over the set $N$ of possible data
types is critical for this upper bound. Say that one type $(C,D)$ dominates
another $(C^{\prime },D^{\prime })$ if $D\supseteq C^{\prime }$ (i.e., the
former type's set of control variables contains the latter type's set of
conditioning variables). When $t$ is constant, the upper bound is $0$ when
the domination relation over $N$ is complete and quasitransitive, and $c$
when it is not.\footnote{%
A binary relation is quasitransitive if its asymmetric part is transitive
(following Sen (1969)).} Thus, when data types are ordered in a particular
sense, the equilibrium condition eliminates all welfare loss due to causal
errors. Conversely, when data types are not ordered, the upper bound on the
DM's welfare loss is the same as when we do not impose any restriction on
the conditional action distribution. The former case fits situations in
which DM types are \textquotedblleft vertically\textquotedblright\
differentiated, (roughly) in the sense that some types control for more
variables than others. The latter case fits \textquotedblleft
horizontal\textquotedblright\ differentiation, in the sense that different
types control for different variables.

I obtain partial characterization results when there is variation in $t$ and 
$D=C$ for all data types. When the domination relation is complete and the $y
$ is purely a function of $t$, the upper bound on the DM's equilibrium
welfare loss is $\Pr (t=1)\cdot \Pr (t=0)$. When the relation is incomplete,
the upper bound is at least $\max \{\Pr (t=1),\Pr (t=0)\}$. When in addition
there is no restriction on the conditional outcome distribution, the upper
bound is $1$. Once again, the structure of the domination relation over the
set of possible data types plays a key role in how equilibrium forces
constrain the cost of flawed causal inference.\footnote{%
Spiegler (2022) presents an example of how equilibrium forces can restrict
the cost of committing a reverse-causality misperception.}

\section{A Model}

Let $a$, $t$ and $y$ be three binary variables that take values in $\{0,1\}$%
, where: $a$ is an \textit{action} that a decision maker (DM) chooses; $y$
is an \textit{outcome}; and $t$ is the DM's \textit{preference type}. Let $%
x=(x_{1},...,x_{K})$ be a collection of additional exogenous variables,
which are realized jointly with $t$, and prior to the realization of $a$ and 
$y$. Let $A=\{0,1\}$ denote the set of values that $a$ can take. Let $X_{k}$
be the set of values that the variable $x_{k}$ can take. For every $M\subset
\{1,...,K\}$, denote $x_{M}=(x_{k})_{k\in M}$ and $X_{M}=\times _{k\in
M}X_{k}$.

I assume that $x$ and $t$ are the sole potential causes of $y$ --- i.e., $a$
has \textit{no causal effect} on $y$. This assumption is made for
expositional clarity; I will relax it in Section 5.

The DM is a subjective expected utility maximizer, whose vNM utility
function is%
\[
u(t,a,y)=y-c\cdot \mathbf{1}[a\neq t] 
\]%
where $c\in (0,1)$ is a constant. Thus, the DM has an intrinsic motive to
match his action to his preference type; he will choose $a\neq t$ only if he
believes that this increases the probability of the outcome $y=1$. If the DM
understood that $a$ has no causal effect on $y$, he would always choose $a=t$%
.

The DM's \textit{data type} is defined by a pair $(C,D)$, where $C\subseteq
D\subseteq \{1,...,K\}$. The interpretation is that $C$ defines the set of $x
$ variables that the type can condition on, because he observes their
realization before taking an action; and $D$ is the set of exogenous
variables about which he has long-run data (note that $t$ is never among
these variables). There are $n$ data types. Denote $N=\{1,...,n\}$. Each
data type $i\in N$ is associated with a distinct pair $(C_{i},D_{i})$. We
say that type $i$ is $simple$ if $C_{i}=D_{i}$ --- i.e., the DM only has
long-run data about the variables he conditions on. Let $\lambda \in \Delta
(N)$ be a prior distribution over data types. This distribution is \textit{%
independent} of all other variables. A strategy for type $(t,i)$ is a
function $\sigma _{t,i}:X\rightarrow \Delta (A)$. By definition, this
strategy is measurable with respect to $X_{C_{i}}$.

Let $p$ be a joint long-run probability distribution over $t,x,a,y$. Denote $%
\gamma =p(t=1)$. The assumption that $a$ has no causal effect on $y$ means
that $p$ satisfies the conditional-independence property $y\perp a\mid (t,x)$%
.\footnote{%
Throughout the paper, I use the symbol $\perp $ to denote statistical
independence.} The distribution $p$ can thus be factorized as follows:%
\[
p(t,x,a,y)=p(t,x)p(a\mid t,x)p(y\mid t,x) 
\]%
where the term $p(a\mid t,x)$ represents the DM's average behavior across
data types:%
\[
p(a\mid t,x)=\sum_{i\in N}\lambda _{i}\sigma _{t,i}(a\mid t,x_{C_{i}}) 
\]%
This term is endogenous, whereas $p(t,x)$ and $p(y\mid t,x)$ are exogenous.

I assume that a DM of data type $i$ forms the following belief regarding the
causal effect of $a$ on $y$ given his observation of $x_{C_{i}}$:%
\begin{equation}
\tilde{p}_{i}(y\mid x_{C_{i}},do(a))=\sum_{x_{D_{i}\backslash
C_{i}}}p(x_{D_{i}}\mid x_{C_{i}})p(y\mid a,x_{D_{i}})
\label{subjective belief}
\end{equation}%
The $\tilde{p}$ notation indicates that this is a subjective belief, which
may be incorrect. The $do$ notation follows Pearl (2009). Its role here is
merely to indicate that (\ref{subjective belief}) is a causal quantity, to
be distinguished from purely probabilistic conditioning. The DM's attempt to
evaluate the causal effect of $a$ on $y$ impels him to \textit{control }for
every exogenous variable about which he has data. For some of these
variables (represented by $C_{i}$), he also learns their realization prior
to taking his action, and therefore he \textit{conditions} on them. For the
other variables (represented by $D_{i}\backslash C_{i}$), he has data about
their long-run correlation with $a$, $x_{C_{i}}$ and $y$, yet he does not
learn their realization prior to taking an action, and therefore he \textit{%
adjusts} his belief by summing over them.\bigskip

\begin{definition}
Data type $i$'s perceived causal effect of switching from $a=0$ to $a=1$
given $x$ is%
\[
\Delta _{i}(x)=\tilde{p}_{i}(y=1\mid x_{C_{i}},do(a=1))-\tilde{p}%
_{i}(y=1\mid x_{C_{i}},do(a=0))\smallskip 
\]
\end{definition}

Plugging (\ref{subjective belief}) into this definition, we obtain:

\begin{equation}
\Delta _{i}(x)=\sum_{x_{D_{i}\setminus C_{i}}}p(x_{D_{i}\setminus C_{i}}\mid
x_{C_{i}})[p(y=1\mid a=1,x_{D_{i}})-p(y=1\mid a=0,x_{D_{i}})]
\label{delta i}
\end{equation}%
This formula will serve us throughout this paper.

If the DM had long-run data about all exogenous variables (including $t$),
then he could control for all of them. Doing so, he would correctly infer a
null causal effect of $a$ on $y$. In contrast, the DM in this model may end
up believing that $a$ has a non-zero causal effect on $y$ because he fails
to control for all the exogenous variables. In this case, he misinterprets
part of the correlation between $a$ and $y$ as a causal effect, whereas in
reality this correlation is entirely due to confounding by $t,x$.

The preceding paragraph may give the impression that the only case of
\textquotedblleft bad controls\textquotedblright\ that the model captures is 
\textit{insufficient} controls. However, note that while controlling for all 
$K+1$ exogenous variables is always correct, it is possible that a strict
subset of these variables is a sufficient set of controls. In this case,
controlling for additional variables may induce errors, as in the example by
Cinelli et al. (2022) described in the Introduction. Thus, the present model
allows for both insufficient and excessive controlling. However, the model
does not accommodate variables that are caused by $a$ or $y$ as possible
controls --- it only focuses on so-called \textquotedblleft
pre-treatment\textquotedblright\ variables.\bigskip

\begin{definition}
Let $\varepsilon >0$. A strategy profile $\sigma =(\sigma _{1},...,\sigma
_{n})$ is an $\varepsilon $-equilibrium if for every $i=1,...,n$ and every $%
t,x,a^{\prime }$, $\sigma _{i}(a^{\prime }\mid t,x)>\varepsilon $ only if%
\[
a^{\prime }\in \arg \max_{a}\sum \tilde{p}_{i}(y\mid
x_{C_{i}},do(a))u(t,a,y) 
\]%
An equilibrium is a limit of a sequence of $\varepsilon $-equilibria for $%
\varepsilon \rightarrow 0$.\bigskip
\end{definition}

The trembling-hand aspect of the equilibrium concept will not play a role in
the characterization results of Section 3-4, except for one result.

The structure of $u$ means that in equilibrium, type $i$ will play $a\neq t$
with positive probability at $x$ only if%
\[
\left\vert \Delta _{i}(x)\right\vert \geq c 
\]%
Since $a$ has no causal effect on $y$, playing $a\neq t$ yields a welfare
loss.\bigskip

\begin{definition}[Expected welfare loss]
Given a strategy profile $\sigma $, the DM's expected welfare loss is%
\[
c\sum_{t,x}p(t,x)\dsum\limits_{i\in N}\lambda _{i}\sigma _{i}(a=1-t\mid
t,x)\smallskip 
\]
\end{definition}

Since the welfare loss is proportional to $c$, I will subsequently fix $c$
and quantify the welfare loss by the probability of playing $a\neq t$. My
main task in the next sections will be to derive upper bounds on this
quantity when $\sigma $ is required to be an equilibrium. Without this
equilibrium condition, the upper bound is $1$. To see why, suppose that $t=0$
with certainty, and that $x\in \{0,1\}$. Assume $y=x$ with probability one
for every $x$, and consider the strategy $\sigma $ that prescribes $a=x$
with probability one. Then, by definition, the probability of error is one.
However, the strategy $\sigma $ is inconsistent with equilibrium. For the DM
to vary $a$ with $x$, he must be able to \textit{condition} on $x$ --- i.e., 
$C_{i}\neq \emptyset $. But this means the DM correctly \textit{controls}
for $x$ when estimating the causal effect of $a$ on $y$, which means that he
correctly estimates it to be zero, contradicting the assumption that he
plays $a\neq t$ for some realization of $x$. It follows that the requirement
that $\sigma $ is an equilibrium strategy can have bite.\bigskip

\noindent \textit{Comment: Why does }$C\subseteq D$? The assumption that $%
C\subseteq D$ means that if a DM conditions on a variable, he must have
long-run data about it. In principle, one can easily imagine situations in
which agents know the realization of a variable without having data about
the long-run behavior of this variable. For instance, the DM may know his
height but lack access the statistics about how height is correlated with
the outcome of interest. In the absence of such data, the DM cannot make use
of his height information, and therefore, we might as well assume that he
lacks it. This is the justification for the assumption that $C\subseteq D$.
Note that the DM knows the realization of $t$, and he makes use of this
information to calculate his utility, but this does not require access to
any long-run statistical data.\bigskip 

\noindent \textit{Comment: A \textquotedblleft persuasion\textquotedblright\
interpretation}. The worst-case analysis can be interpreted through the
prism of the small literature on persuading boundedly rational agents (e.g.,
Glazer and Rubinstein (2012), Galperti (2019), Hagenbach and Koessler
(2020), Schwartzstein and Sunderam (2021), Eliaz et al. (2021), and De
Barreda et al. (2022)). Under this interpretation, the DM is the receiver
who takes an action. The sender's objective is to maximize the probability
that the receiver plays $a\neq t$. Toward this end, he designs two features
of the receiver's environment. The conventional feature is a distribution
over the receiver's signals. The less conventional feature (but one that is
closer in spirit to Eliaz et al. (2021)) involves the long-run statistical
data to which the receiver has access, according to which he forms his
beliefs. Worst-case analysis can thus be viewed as finding the sender's
optimal data provision strategy.

\section{Analysis: Homogenous Preferences}

In this section I characterize the maximal welfare loss that is consistent
with equilibrium behavior, when there is no variation in the DM's
preferences. Specifically, assume that $t=0$ with probability one, such that
the DM's expected welfare loss is simply $c$ times the ex-ante probability
that he plays $a=1$. I show that the upper bound on this probability depends
on a simple property of the set of possible data types.

In this environment of preference homogeneity, the only potential source of
variation in the DM's behavior is the way the various types condition their
actions on $x$. Therefore, for any set $N$ of data types, there is an
equilibrium in which the DM plays $a=0$ with probability one. To see why,
construct the following sequence of perturbations around this strategy: for
every $\varepsilon \in (0,\frac{1}{2})$, every data type $i$ plays $a=1$
with probability $\varepsilon $, independently of $x_{C_{i}}$. By
construction, $a\perp x$ under this strategy profile, and therefore $\Delta
_{i}(x)=0$ for every type $i$, such that $a=0$ is the type's unique
best-reply. The question is whether there are additional equilibria, in
which the DM commits an error with positive probability, and how large this
probability can get. The following example serves to illustrate this problem.%
$\bigskip $

\noindent \textit{Example 3.1}

\noindent Let $K=2$. The two exogenous variables $x_{1}$ and $x_{2}$ take
values in $\{0,1\}$, and their joint distribution satisfies:%
\begin{eqnarray*}
p(x_{1} &=&1)=p(x_{2}=1)=\beta \in (0,1)\smallskip \\
p(x_{2} &=&1\mid x_{1}=1)=p(x_{1}=1\mid x_{2}=1)=q\in \lbrack \frac{1}{2}%
,1)\smallskip \\
p(y &=&1\mid x_{1},x_{2})=x_{1}x_{2}\text{ for every }x_{1},x_{2}\text{.}
\end{eqnarray*}%
Let $n=2$, $\lambda _{1}=\lambda _{2}=\frac{1}{2}$, where $C_{i}=D_{i}=\{i\}$%
. That is, each type conditions his action on a different aspect of $x$.

The following is an interpretation of this specification. The DM is an
executive who chooses a business strategy for a company whose environment is
defined by financial and technological factors (represented by $x_{1}$ and $%
x_{2}$). The company is profitable if both factors are favorable. The
executive's decision is informed by an analyst's report. There are two types
of analysts, who specialize in (and therefore monitor) the technological and
financial environments, respectively.

Suppose that each type $i=1,2$ always plays $a_{i}=x_{i}$. Let us examine
whether this strategy profile is an equilibrium. Begin by calculating type $%
1 $'s subjective estimate of actions' causal effect on profits, given his
information. First, observe that since $y=x_{1}x_{2}$ independently of $a$,%
\begin{eqnarray*}
p(y &=&1\mid a,x_{1}=1)=p(x_{2}=1\mid a,x_{1}=1) \\
p(y &=&1\mid a,x_{1}=0)=0
\end{eqnarray*}%
for every $a$. (Note that these quantities never involve conditioning on a
zero-probability event. For example, the combination $a=0,x_{1}=1$ occurs
when $x_{2}=0$ and the DM is of type $2$.) Therefore, we only need to
calculate the following conditional probabilities, which also make use of
the DM's postulated strategy:%
\begin{eqnarray*}
p(x_{2} &=&1\mid a=1,x_{1}=1)=\frac{q}{q+\frac{1}{2}(1-q)}\medskip \\
p(x_{2} &=&1\mid a=0,x_{1}=1)=0
\end{eqnarray*}%
It follows that%
\[
\Delta _{1}(x_{1}=1)=\frac{q}{q+\frac{1}{2}(1-q)}-0=\frac{2q}{1+q} 
\]%
Therefore, if $2q/(1+q)>c$, type $1$ will prefer to play $a=1$ when $x_{1}=1$%
. In addition, we established that%
\[
\Delta _{1}(x_{1}=0)=0-0=0 
\]%
Therefore, type $1$ will prefer to play $a=0$ when $x_{1}=0$.

The same calculation applies to type $2$. It follows that as long as $%
q>c/(2-c)$, the postulated strategy profile is an equilibrium. The
equilibrium error probability (i.e., $\Pr (a=1)$) is $\beta $, which can be
arbitrarily close to one --- hence, the equilibrium welfare loss can be as
large as the non-equilibrium benchmark. Thus, unlike Example 1.1, here
equilibrium forces do not \textquotedblleft protect\textquotedblright\ DMs
from their errors of causal inference.

The intuition behind this result is that since each type conditions on a
different component of $x$, each creates a confounding effect that
\textquotedblleft fools\textquotedblright\ the other type. For example, type 
$1$ is vulnerable to interpreting the residual correlation between $a$ and $%
y $ after controlling for $x_{1}$ --- which exists because of type $2$'s
behavior --- as a causal effect. Note that the result does $not$ necessitate
correlation between $x_{1}$ and $x_{2}$. Indeed, even when $q=\frac{1}{2}$,
the above equilibrium can be sustained as long as $c<\frac{2}{3}$. The
reason is that although the DM types in this case condition their actions on
independent exogenous variables, their subjective causal estimates involve
conditioning on $a$. Since this variable is a common consequence of $x_{1}$
and $x_{2}$, conditioning on it creates correlation between otherwise
independent variables.

The equilibrium welfare loss is non-monotone with respect to the data types'
sets of control variables. For example, suppose that type $C_{1}=\{1\}$ and $%
C_{2}=D_{2}=\emptyset $ --- i.e., type $2$ now does not control for any
variable. By definition, he does not vary his action with $x$, and therefore 
$x_{2}$ is not a confounding variable. This means that type $1$ effectively
controls for any potential confounder, and therefore he will not commit any
error in equilibrium. This is similar to Example 1.1, where equilibrium
forces eliminated any welfare loss due to faulty causal inferences. $\square 
$\bigskip

The example shows that for some sets of data types, the equilibrium welfare
loss is zero, while for others, it can be as large as when we do not impose
any equilibrium restriction. The results in this section generalize this
lesson. They will make use of the following binary relation $P$ over data
types.\bigskip

\begin{definition}
For data types $i,j\in N$, $iPj$ if $D_{i}\supseteq C_{j}$.\bigskip
\end{definition}

The meaning of $iPj$ is that the set of variables that data type $i$
controls for (via conditioning or adjustment) is a weak superset of the set
of variables that type $j$ conditions on. Note that by our definition of
types, $P$ is reflexive, since $D_{i}\supseteq C_{i}$ for every $i\in N$.
Let $P^{\ast }$ be the asymmetric (strict) part of $P$ --- i.e., $iP^{\ast
}j $ if $iPj$ and $j\NEG{P}i$. Following Sen (1969), $P$ is \textit{%
quasitransitive} if $P^{\ast }$ is transitive.\bigskip

\begin{lemma}
\label{lemma}Suppose a binary relation $P$ over $N$ is complete and
quasitransitive. Then, $N$ can be partitioned into $L$ classes, $%
N_{1},...,N_{L}$, as follows: For every $\ell =1,...,L$,%
\[
N_{\ell }=\{i\notin \cup _{h<\ell }N_{h}\mid j\NEG{P}^{\ast }i\text{ for all 
}j\notin \cup _{h<\ell }N_{h}\} 
\]%
Moreover, for every $i\in N_{\ell }$, $iPj$ for all $j\in \cup _{h\geq \ell
}N_{h}$.
\end{lemma}

\begin{proof}
By definition, $P^{\ast }$ does not contain cycles. Hence, the set of data
types $i\in N$ such that $j\NEG{P}^{\ast }i$ for all $j\in N$ (i.e., the set
of $P^{\ast }$-undominated data types) is non-empty. Define this set by $%
N_{1}$. Since $P$ is complete, $iPj$ for every $i\in N_{1}$ and every $j\in
N $. The other cells in the partition are defined inductively: After $%
N_{1},...,N_{\ell }$ are removed from $N$, let $N_{\ell +1}$ be the set of $%
P^{\ast }$-undominated types in the remaining set. Since none of the sets $%
N_{\ell }$ is empty, the procedure terminates after at most $n$
steps.\bigskip
\end{proof}

The lemma confirms that when $P$ is complete and quasitransitive, it
partitions $N$ into layers, such that the first (top) layer consists of all $%
P^{\ast }$-undominated, the second layer consists of all $P^{\ast }$%
-undominated element outside the first layer, and so forth.

When all data types are $simple$, the structure of $P$ is simplified: $iPj$
means $C_{i}\supseteq C_{j}$, hence $P$ is automatically transitive. The
relevant distinction in this case is thus between complete and incomplete $P$%
. Moreover, since I assumed that all data types are distinct, it follows
that for every pair of distinct types $i,j$, $iPj$ implies $iP^{\ast }j$.
Therefore, under simple data types, the requirement that $P$ is complete is
reduced to the requirement that $P^{\ast }$ is a \textit{linear ordering}.

The following results fully characterize the maximal equilibrium welfare
loss, as a function of $P$. The first result generalizes Example 1.1,
whereas the second result generalizes Example 3.1.\bigskip

\begin{proposition}
\label{prop complete}Let $\gamma =0$. Suppose that $P$ is complete and
quasitransitive. Then, the DM's expected welfare loss is zero in any
equilibrium.
\end{proposition}

\begin{proof}
I will show that $a=0$ with probability one in any equilibrium. The proof is
by induction with respect to the partition defined by Lemma \ref{lemma}.
Consider an arbitrary type $i$ in the top layer $N_{1}$. This type satisfies 
$D_{i}\supseteq C_{j}$ for all $j\in N$. Hence, there is no $x$ variable
outside $D_{i}$ that $any$ DM type conditions his action on. This means that 
$y\perp a\mid x_{D_{i}}$ --- i.e., $p(y=1\mid a,x_{D_{i}})=p(y=1\mid
x_{D_{i}})$. Therefore, $\Delta _{i}(x)=0$. It follows that in equilibrium,
type $i$ plays $a=0$ for all $x$.

Suppose the claim holds for all types in the top $m$ layers in the
partition, and now consider an arbitrary type $i$ in the $(m+1)$-th layer.
By definition, $D_{i}\supseteq C_{j}$ for every type $j$ outside the top $m$
layers of the partition. As to types in the top $m$ layers, by the inductive
step\ these types play a constant action $a=0$ for all $x$ in any
equilibrium --- i.e., there is no variation in their action. It follows that
if $p$ is consistent with equilibrium, then $y\perp a\mid x_{D_{i}}$.
Formula (\ref{delta i}) then implies $\Delta _{i}(x)=0$. It follows that in
equilibrium, type $i$ plays $a=0$ for all $x$. This completes the inductive
proof.\bigskip
\end{proof}

Thus, when $\gamma =0$ and the binary relation $P$ is complete and
quasitransitive --- i.e., the data types are ordered in a certain sense ---
the equilibrium requirement fully \textquotedblleft
protects\textquotedblright\ the DM from choice errors due to flawed causal
inference. It does so by shutting down the channels through which the choice
behavior of some types could confound the relation between other types'
actions and $y$. Types in the top layer of the $P$-based partition
effectively control for all sources of correlation between $a$ and $y$. Even
when a top-layer type does not control for some exogenous variable, this
does not matter because no other type conditions on this variable, hence it
generates no confounding effect. As a result, top-layer types' subjective
best-replying implies that they do not generate any variation in choice
behavior. This means that types in the next layer effectively control for
all potential confounders --- which would not be the case if we did not
impose the equilibrium condition on the behavior of top-layer types. This
equilibrium effect spreads through all layers of the partition.\bigskip

\begin{proposition}
\label{prop incomplete}Let $\gamma =0$. Suppose that $P$ violates
completeness or quasitransitivity. Then, for any $c,\beta \in (0,1)$, there
exist $\lambda $ and $(p(x,y))$ such that $\Pr (a=1)>\beta $ in some
equilibrium.
\end{proposition}

\begin{proof}
Suppose first that $P$ is incomplete. Then, there exist two types, denoted
conveniently $1$ and $2$, such that $C_{1}\setminus D_{2}$ and $%
C_{2}\setminus D_{1}$ are non-empty. Select two variables in $C_{1}\setminus
D_{2}$ and $C_{2}\setminus D_{1}$, and denote them $1$ and $2$ as well,
respectively. Suppose that $\lambda _{1},\lambda _{2}>0$ and\ $\lambda
_{1}+\lambda _{2}=1$. Construct $p$ as follows. First, let $x_{1},x_{2}\in
\{0,1\}$, and%
\begin{eqnarray*}
p(x_{1} &=&1,x_{2}=1)=1-\varepsilon \\
p(x_{1} &=&0,x_{2}=1)=p(x_{1}=1,x_{2}=0)=\frac{\varepsilon }{2}
\end{eqnarray*}%
where $\varepsilon >0$ is arbitrarily small. Second, let $p(y=1\mid
x_{1},x_{2})=x_{1}x_{2}$. Thus, $x_{1}$ and $x_{2}$ are the only $x$
variables that determine $y$, and so we can afford to ignore all other $x$
variables.

Given this specification of $\lambda $ and $p(x,y)$, let us now construct an
equilibrium in which for each type $i=1,2$, $a_{i}=x_{i}$ with probability
one. Without loss of generality, consider type $1$'s reasoning. This type's
perceived causal effect of $a$ on $y$ given $x_{1}$ is%
\[
\Delta _{1}(x_{1})=p(y=1\mid a=1,x_{1})-p(y=1\mid a=0,x_{1}) 
\]%
because all other variables are either not in $D_{1}$ or irrelevant for the
determination of $y$ and therefore can be ignored. Note that since $%
y=x_{1}x_{2}$ with probability one,%
\begin{eqnarray*}
p(y &=&1\mid a,x_{1}=1)=p(x_{2}=1\mid a,x_{1}=1) \\
p(y &=&1\mid a,x_{1}=0)=0
\end{eqnarray*}%
for every $a$. Neither of these two formulas ever conditions on null events,
because all four combinations of $(a,x_{1})$ occur with positive
probability. (Note that the combination $a=1$ and $x_{1}=0$ arises when $%
x_{2}=1$, because of the DM's strategy.) By our construction of $%
p(x_{1},x_{2})$ and the DM's strategy,%
\begin{eqnarray*}
p(x_{2} &=&1\mid a=1,x_{1}=1)=\frac{1-\varepsilon }{1-\varepsilon +\frac{%
\varepsilon }{2}\cdot \lambda _{1}} \\
p(x_{2} &=&1\mid a=0,x_{1}=1)=0
\end{eqnarray*}%
It follows that $\Delta _{1}(x_{1}=0)=0$, while $\Delta
_{1}(x_{1}=1)\rightarrow 1$ as $\varepsilon \rightarrow 0$. Therefore, type $%
1$'s postulated strategy can be consistent with equilibrium for values of $c$
that are arbitrarily close to one.

Now suppose that $P$ is complete but not quasitransitive. This means that $%
P^{\ast }$ must have a cycle of length $3$ --- that is, we can find three
types, denoted $1,2,3$, such that $1P^{\ast }2$, $2P^{\ast }3$ and $3P^{\ast
}1$. By the definition of $P^{\ast }$, this means that for each of the three
types $i=1,2,3$, there is a distinct variable in $\{1,...,K\}$, conveniently
denoted $i$ as well, such that $1\in C_{1}\setminus D_{2}$, $2\in
C_{2}\setminus D_{3}$ and $3\in C_{3}\setminus D_{1}$. Suppose $\lambda
_{1},\lambda _{2},\lambda _{3}>0$ and $\lambda _{1}+\lambda _{2}+\lambda
_{3}=1$. Let $x_{1},x_{2},x_{3}\in \{0,1\}$. Construct $p$ as follows: First,%
\[
p(x_{1}=1,x_{2}=1,x_{3}=1)=1-\varepsilon 
\]%
and%
\[
p(x_{i}=0,x_{j}=x_{k}=1)=\frac{\varepsilon }{3} 
\]%
for every $i=1,2,3$ and $j,k\neq i$, where $\varepsilon >0$ is arbitrarily
small. Second, let $p(y=1\mid x_{1},x_{2},x_{3})=x_{1}x_{2}x_{3}$. Thus, $%
x_{1},x_{2},x_{3}$ are the only $x$ variables that determine $y$, and so we
can afford to ignore all other $x$ variables. Suppose each type $i=1,2,3$
plays $a=x_{i}$ with probability one. Using essentially the same calculation
as in the case of incomplete $P$, we can see that for every $i=1,2,3$, $%
\Delta _{i}(x_{i}=0)=0$, whereas $\Delta _{i}(x_{i}=1)\rightarrow 1$ as $%
\varepsilon \rightarrow 0$. Therefore, the postulated strategy profile is an
equilibrium.\bigskip
\end{proof}

Thus, the upper bound on the DM's equilibrium welfare loss due to wrong
causal inferences critically depends on whether the binary relation $P$ is
complete and quasitransitive. When it is, the equilibrium behavior of some
data types cannot generate a variation that produces confounding patterns
that other data types misinterpret as causal. When it is not, the
equilibrium behavior of different types can create such confounding patterns
that mutually sustain their causal-inference errors. In that case, the
equilibrium assumption does not constrain the maximal possible welfare loss
due to these errors.

The distinction between the two cases can be described as a distinction
between \textquotedblleft vertical\textquotedblright\ and \textquotedblleft
horizontal\textquotedblright\ differentiation among data types. This is
especially palpable in the case of simple types, where the results hinge on
whether $P^{\ast }$ is a linear ordering. When it is, the types' sets of
control variables are ordered by set inclusion, and in this case the
equilibrium welfare loss is zero. When it is not, the difference between
types is that they control for different variables, and this
\textquotedblleft horizontal\textquotedblright\ differentiation enables them
to create mutually reinforcing confounding patterns.

\section{Analysis: Heterogeneous Preferences}

In this section I reintroduce preference heterogeneity, by assuming $\gamma
\in (0,1)$. Unlike the homoegenous-preference case, here I lack a complete
characterization of the maximal equilibrium welfare loss, and present a
number of partial results. In particular, I restrict attention to \textit{%
simple data types}, as defined in Section 2 --- that is, $C_{i}=D_{i}$ for
every data type $i$. Recall that in this case, $P^{\ast }$ is complete if
and only if it is a linear ordering. Denote $\delta _{t}=p(y=1\mid t)$.
Without loss of generality, assume $\delta _{1}\geq \delta _{0}$.\bigskip

\noindent \textit{Example 4.1}

\noindent Suppose $\delta _{t}\equiv t$. Let $K=0$ and $n=1$. This means
that there is a unique data type, $C=\emptyset $. One interpretation for
this setting is that $a$ represents a student's decision whether to select a
math-intensive major in college; $t$ indicates whether he likes math; and $y$
represents his subsequent earnings. The student learns the correlation
between $a$ and $y$. He has no access to control variables, and therefore
ends up treating the correlation as causal. The assumption that $\delta
_{t}\equiv t$ means that fondness for math is perfectly correlated with math
skills that determine earnings, independently of the student's decision.

I will now show that this setting gives rise to a unique equilibrium, and I
will characterize the DM's expected welfare loss in this equilibrium.

First, let us obtain an expression for the DM's estimated causal effect of $a
$ on $y$. By definition, it is%
\[
\Delta =p(y=1\mid a=1)-p(y=1\mid a=0)
\]%
Denote $\alpha _{t}=\sigma (a=1\mid t)$. When the DM's strategy is fully
mixed, $\alpha _{t}\in (0,1)$ for every $t$. By the DM's preferences, $%
\alpha _{1}\geq \alpha _{0}$. Now obtain explicit expressions for the terms
that define $\Delta $:%
\begin{eqnarray*}
p(y &=&1\mid a=1)=\frac{\gamma \cdot \alpha _{1}\cdot \delta _{1}+(1-\gamma
)\cdot \alpha _{0}\cdot \delta _{0}}{\gamma \cdot \alpha _{1}+(1-\gamma
)\cdot \alpha _{0}}\medskip  \\
p(y &=&1\mid a=0)=\frac{\gamma \cdot (1-\alpha _{1})\cdot \delta
_{1}+(1-\gamma )\cdot (1-\alpha _{0})\cdot \delta _{0}}{\gamma \cdot
(1-\alpha _{1})+(1-\gamma )\cdot (1-\alpha _{0})}
\end{eqnarray*}%
A simple calculation establishes that since $\delta _{1}=1>0=\delta _{0}$
and $\alpha _{1}\geq \alpha _{0}$, we must have $\Delta \geq 0$. This in
turn implies that $\alpha _{1}\geq 1-\varepsilon $ in $\varepsilon $%
-equilibrium, because when $t=1$, the DM perceives no conflict between his
intrinsic taste for $t=1$ and the effect of his choice on $y$. Plugging the
known expressions for $\alpha _{1}$ and $\delta _{t}$ and taking the $%
\varepsilon \rightarrow 0$ limit, we obtain%
\[
\Delta =\frac{\gamma }{\gamma +(1-\gamma )\cdot \alpha _{0}}
\]

If $\alpha _{0}\leq \varepsilon $ in $\varepsilon $-equilibrium, then $%
\Delta \rightarrow 1$ in the $\varepsilon \rightarrow 0$ limit. But then, $%
\Delta >c$, hence playing $a=1$ at $t=0$ is subjectively optimal, in
contradiction with $\alpha _{0}\leq \varepsilon $. It follows that $\alpha
_{0}>0$ in equilibrium. There are two cases to consider. First, suppose $%
\alpha _{0}\in (0,1)$. This requires $\Delta =c$ (and therefore $\gamma <c$%
), such that%
\[
\alpha _{0}=\frac{\gamma (1-c)}{(1-\gamma )c} 
\]%
Since the DM only commits an error in equilibrium when $t=0$, his expected
equilibrium welfare loss is%
\[
c\cdot (1-\gamma )\cdot \alpha _{0}=\gamma (1-c)<\gamma (1-\gamma ) 
\]%
By setting $\gamma \approx c$, we can get arbitrarily close to the upper
bound of $\gamma (1-\gamma )$.

Second, suppose $\alpha _{0}=1$. This requires us to sustain this
equilibrium with suitable trembles. Specifically, suppose that such that $%
\alpha _{1}=1-\varepsilon ^{2}$ whereas $\alpha _{0}=1-\varepsilon $. Then,
as $\varepsilon \rightarrow 0$, we obtain $p(y=1\mid a=1)\approx \gamma $
and $p(y=1\mid a=1)\approx 0$. If $\gamma >c$, this is consistent with
equilibrium. The DM's welfare loss in this equilibrium is%
\[
c\cdot (1-\gamma )\cdot 1<\gamma (1-\gamma ) 
\]%
Again, by setting $\gamma \approx c$, we can get arbitrarily close to this
upper bound. $\square $\bigskip

Thus, for any configuration of $c$ and $\gamma $, there is a unique
equilibrium in this setting. The DM's equilibrium welfare loss in this
equilibrium is always below $\gamma (1-\gamma )$. This bound can be
approximated arbitrarily well by setting $\gamma \approx c$. Furthermore,
the trembling-hand aspect of our equilibrium concept is not necessary for
the upper bound.

As in earlier examples, equilibrium forces in Example 3.1 \textquotedblleft
protect\textquotedblright\ the DM against causal errors, by limiting his
welfare loss below $\gamma (1-\gamma )$ --- compared with the
non-equilibrium benchmark of $1$. The intuition is as follows. The DM
mistakes the correlation between $a$ and $y$ for a causal effect. This
correlation is large when $a$ varies strongly with $t$; it hits the maximal
level when $a$ always coincides with $t$. However, that extreme case is
precisely when the DM commits no error. At the other extreme, if the DM
almost always plays $a=1$ because his estimated causal effect of $a$ on $y$
is above $c$, the frequency of the DM's error is maximal. However, since in
this case $a$ varies little with $y$, the estimated causal effect is
smaller. In general, a larger estimated causal effect implies a lower
equilibrium frequency of making an error. This is why equilibrium behavior
limits the DM's expected welfare loss due to wrong causal inference. $%
\square $\bigskip

Let us now turn to characterizations of the upper bound on the DM's
equilibrium welfare loss, under certain restrictions on the data-generating
process. I begin by imposing the domain restriction that $p(y\mid t,x)\equiv
p(y\mid t)$ --- i.e., $y\perp x\mid t$. This fits situations in which the
DM's preference type is a sufficient statistic for determining the outcome,
and the $x$ variables are only potential correlates of this statistic. For
instance, whether a student regards studying as a costly or pleasurable
activity is the cause of her school performance. This attitude (which is not
observable to others) may be correlated with observable socioeconomic
indicators, but these are merely proxies for the true cause.\bigskip 

\begin{proposition}
\label{prop bound linear ordering}Suppose that all data types are simple and
that $P$ is complete. If $y\perp x\mid t$, then the DM's expected welfare
loss in equilibrium is at most $\gamma (1-\gamma )$.
\end{proposition}

\begin{proof}
The proof proceeds stepwise. Recall that since $P$ is complete, $P^{\ast }$
is a linear ordering. For convenience, enumerate the types according to $%
P^{\ast }$ --- i.e., $1P^{\ast }2P^{\ast }\cdots P^{\ast }n$. For every $x$
and every $C\subseteq \{1,...,K\}$, denote $\gamma (x)=p(t=1\mid x)$ and $%
\gamma (x_{C})=p(t=1\mid x_{C})$.\bigskip

\noindent \textbf{Step 1}: Deriving an expression for $\Delta _{i}(x)$

\noindent \textbf{Proof}: Since $y\perp (a,x)\mid t$, we can write%
\[
p(y\mid a,x_{C_{i}})=\sum_{t}p(t\mid a,x_{_{C_{i}}})p(y\mid
a,x_{C_{i}},t)=\sum_{t}p(t\mid a,x_{C_{i}})p(y\mid t) 
\]%
Plugging this in (\ref{delta i}), we obtain%
\begin{equation}
\Delta _{i}(x)=[p(t=1\mid a=1,x_{C_{i}})-p(t=1\mid a=0,x_{C_{i}})][\delta
_{1}-\delta _{0}]\bigskip  \label{delta y(t) case}
\end{equation}

\noindent \textbf{Step 2}: For every $x$, $\Delta _{1}(x)\geq 0$ and $\sigma
_{1}(a=1\mid t=1,x_{C_{1}})=1$.

\noindent \textbf{Proof}: For every $a$, the terms $p(t=1\mid a,x_{C_{i}})$
in (\ref{delta y(t) case}) can be written as%
\begin{equation}
\frac{\gamma (x_{C_{i}})p(a\mid t=1,x_{C_{i}})}{\gamma (x_{C_{i}})p(a\mid
t=1,x_{C_{i}})+(1-\gamma (x_{C_{i}}))p(a\mid t=0,x_{C_{i}})}
\label{p(t=1|a,x)}
\end{equation}%
We will now focus on the term $p(a\mid t=1,x_{C_{1}})$. Note that 
\begin{equation}
p(a\mid t,x_{C_{1}})=\sum_{x_{-C_{1}}}p(x_{-C_{1}}\mid t,x_{C_{1}})p(a\mid
t,x_{C_{1}},x_{_{-C_{1}}})  \label{p(a|t,xCi)}
\end{equation}%
By definition, $C_{1}\supset C_{j}$ for every $j>1$. This means that no data
type $j$ conditions his actions on $x_{-C_{1}}$. Therefore, (\ref{p(a|t,xCi)}%
) is equal to%
\[
\sum_{j=1}^{n}\lambda _{j}\sigma _{j}(a\mid t,x_{C_{j}}) 
\]%
By the DM's preferences, $\sigma _{i}(a=1\mid t=1,x_{C_{i}})\geq \sigma
_{i}(a=1\mid t=0,x_{C_{i}})$ in any equilibrium, for every $i$ and every $x.$
It follows that $p(a=1\mid t=1,x_{C_{1}})\geq p(a=1\mid t=0,x_{C_{1}})$ for
every $x_{C_{1}}$. A simple calculation then confirms that the expression (%
\ref{p(t=1|a,x)}) is weakly increasing in $a$ for $i=1$. Since $\delta
_{1}-\delta _{0}\geq 0$, it follows that $\Delta _{1}(x)\geq 0$.\bigskip

\noindent \textbf{Step 3}: Extending the property of Step 2 to all data types

\noindent \textbf{Proof}: The proof is by induction on $P^{\ast }$. Suppose
that for every type $j=1,...,m$, $\Delta _{j}(x)\geq 0$ and $\sigma
_{j}(a=1\mid t=1,x_{C_{j}})=1$. Now consider type $i=m+1$. We can write%
\[
p(a\mid t,x_{C_{i}})=\sum_{x_{-C_{i}}}p(x_{-C_{i}}\mid t,x_{C_{i}})\left[
\sum_{j\leq m}\lambda _{j}\sigma _{j}(a\mid t,x_{C_{j}})+\sum_{j>m}\lambda
_{j}\sigma _{j}(a\mid t,x_{C_{j}})\right] 
\]%
By the inductive step, $\sigma _{j}(a=1\mid t=1,x_{C_{j}})=1$ for every $%
j\leq m$. By definition, $C_{j}\subseteq C_{i}$ for every $j\geq m+1$, hence 
$\sigma _{j}(a\mid t,x_{C_{j}})$ is constant in $x_{-C_{i}}$. We already
observed that $\sigma _{j}(a=1\mid t=1,x_{C_{j}})\geq \sigma _{j}(a=1\mid
t=0,x_{C_{j}})$ for every $x_{C_{j}}$. It follows that $p(a=1\mid
t=1,x_{C_{i}})\geq p(a=1\mid t=0,x_{C_{i}})$. As in the proof of Step 2,
applying this inequality to (\ref{p(t=1|a,x)}) implies that $\Delta
_{i}(x)\geq 0$ and $\sigma _{i}(a=1\mid t=1,x_{C_{i}})=1$. This completes
the inductive proof.\bigskip

\noindent \textbf{Step 4}: An upper bound on the expected equilibrium
welfare loss given $x$

\noindent \textbf{Proof}: We have established that in any equilibrium, all
data types play $a=1$ with probability one when $t=1$. Therefore, they only
commit an error if they play $a=1$ with positive probability when $t=0$. Fix
the realization of $x$. Let $i(x)$ be the lowest-indexed type $j$ for which $%
\sigma _{j}(a=1\mid t=0,x_{C_{j}})>0$. Then, the DM's expected welfare loss
given $x$ is%
\[
c(1-\gamma (x))\sum_{j=i(x)}^{n}\lambda _{j}\sigma _{j}(a=1\mid
t=0,x_{C_{j}}) 
\]%
In order for type $i(x)$ to play $a=1$ given $x$ and $t=0$, it must be the
case that $c\leq \Delta _{i(x)}(x)$. By Step 3, $\sigma _{j}(a=1\mid
t=1,x_{C_{j}})=1$ for all $j$, hence $p(a=1\mid t=1,x_{C_{i(x)}})=1$.
Plugging this identity into (\ref{delta y(t) case})-(\ref{p(t=1|a,x)}) and
recalling that $0\leq \delta _{1}-\delta _{0}\leq 1$, we obtain%
\[
\Delta _{i(x)}(x)\leq \frac{\gamma (x_{C_{i(x)}})}{\gamma
(x_{C_{i(x)}})+(1-\gamma (x_{C_{i(x)}}))p(a=1\mid t=0,x_{C_{i(x)}})} 
\]

Since $C_{i}\supseteq C_{j}$ for every $j$ for which $\sigma _{j}(a=1\mid
t=0,x_{C_{j}})>0$, it follows that none of these types $j$ condition on $%
x_{-C_{i(x)}}$. Therefore,%
\[
p(a=1\mid t=0,x_{C_{i(x)}})=\sum_{j=i(x)}^{n}\lambda _{j}\sigma _{j}(a=1\mid
t=0,x_{C_{j}}) 
\]%
Denote this quantity by $\alpha $. This means that the DM's expected welfare
loss given $x$ is at most%
\[
\frac{\gamma (x_{C_{i(x)}})}{\gamma (x_{C_{i(x)}})+(1-\gamma
(x_{C_{i(x)}}))\alpha }\cdot (1-\gamma (x))\cdot \alpha 
\]%
This expression attains its maximal value when $\alpha =1$. Therefore, the
DM's expected welfare loss given $x$ is bounded from above by%
\[
(1-\gamma (x))\gamma (x_{C_{i(x)}})=(1-\gamma (x))\cdot \sum_{x^{\prime
}}p(x^{\prime }\mid x_{C_{i(x)}}^{\prime }=x_{C_{i(x)}})\gamma (x^{\prime
})\bigskip 
\]

\noindent \textbf{Step 5}: Deriving the upper bound on the DM's ex-ante
expected equilibrium welfare loss

\noindent \textbf{Proof}: By Step 4, the ex-ante welfare loss is at most%
\begin{equation}
\sum_{x}p(x)(1-\gamma (x))\cdot \sum_{x^{\prime }}\beta (x^{\prime
},x)\gamma (x^{\prime })  \label{convex}
\end{equation}%
where $\beta (\cdot )$ is a system of convex combinations, $\beta (x^{\prime
},x)=p(x^{\prime }\mid x_{C_{i(x)}}^{\prime }=x_{C_{i(x)}})$. Expression (%
\ref{convex}) is a concave function of $(\gamma (x))_{x}$. By Jensen's
inequality, it attains a maximum when $\gamma (x)=\gamma $ for all $x$, such
that the upper bound on the DM's expected equilibrium welfare loss is $%
\gamma (1-\gamma )$.\bigskip
\end{proof}

Thus, when the set of types is simple and $y$ is only determined by $t$, the
DM's expected equilibrium welfare loss is at most $\gamma (1-\gamma )$.
Example 3.1 established the tightness of this bound. This result also means
that across all distributions that satisfy $y\perp (x,a)\mid t$, the
expected welfare loss is at most $\frac{1}{4}$ --- compared with the
non-equilibrium upper bound of $1$. This is yet another demonstration of how
the equilibrium condition restricts the cost of faulty causal inferences. As 
$\gamma \rightarrow 0$, this loss converges to zero.

When completeness of $P$ is relaxed, finding a tight upper bound on the DM's
expected welfare loss when $y\perp x\mid t$ is an open problem. However, the
following result establishes that this bound cannot be lower than $\max
\{\gamma ,1-\gamma \}$. This carries the relevance of the distinction
between complete and incomplete $P$ to the setting with preference
heterogeneity.\bigskip 

\begin{proposition}
\label{prop sqrt bound}Suppose $P$ is incomplete. Then, for every $\gamma $,
there exist $c$, a distribution $\lambda $ over simple data types and a
distribution $(p(x,y\mid t))$ satisfying $y\perp x\mid t$, for which there
is an equilibrium in which the DM's expected welfare loss is arbitrarily
close to $\max \{\gamma ,1-\gamma \}$.
\end{proposition}

\begin{proof}
Since $P$ is incomplete, $K\geq 2$. Moreover, there exist two data types, $1$
and $2$, and two exogenous variables, conveniently denoted $x_{1}$ and $x_{2}
$, such that $1\in C_{1}\setminus C_{2}$ and $2\in C_{2}\setminus C_{1}$.
Suppose $\lambda _{1}+\lambda _{2}=1$. Without loss of generality, let $%
\gamma \geq \frac{1}{2}$, such that $\max \{\gamma ,1-\gamma \}=\gamma $.
Suppose that $x_{1},x_{2}\in \{0,1,\#\}$, and construct the following
distribution over triples $(t,x_{1},x_{2})$:%
\[
\begin{array}{cccc}
\Pr  & t & x_{1} & x_{2} \\ 
\beta  & 1 & 1 & 1 \\ 
\beta ^{2} & 0 & 1 & 0 \\ 
\beta ^{2} & 0 & 0 & 1 \\ 
1-\gamma  & 0 & \# & \# \\ 
\gamma -\beta -2\beta ^{2} & 1 & 0 & 0%
\end{array}%
\]%
Complete the exogenous components of $p$ by letting $\delta _{1}=1$ and $%
\delta _{0}=0$. Since there are no relevant $x$ variables other than $x_{1}$
and $x_{2}$, we can set without loss of generality $C_{1}=\{1\}$ and $%
C_{2}=\{2\}$.

Suppose each type $i$ plays $a_{i}=x_{i}$ with probability one whenever $%
x_{i}\in \{0,1\}$. In addition, suppose each type $i$ plays $a=0$ with
probability $1-\varepsilon $ when $x_{i}=2$, where $\varepsilon $ and $\beta 
$ are arbitrarily small. Let us calculate the terms in $\Delta _{1}(x_{1}=1)$%
:%
\begin{eqnarray*}
p(t &=&1\mid a=1,x_{1}=1)=\frac{\beta }{\beta +\lambda _{1}\beta ^{2}}%
\approx 1 \\
p(t &=&1\mid a=0,x_{1}=1)=0
\end{eqnarray*}%
such that $\Delta _{1}(x_{1}=1)\approx 1$. Let us now calculate the terms in 
$\Delta _{1}(x_{1}=0)$:%
\begin{eqnarray*}
p(t &=&1\mid a=1,x_{1}=0)=0 \\
p(t &=&1\mid a=0,x_{1}=0)=\frac{\gamma -\beta -2\beta ^{2}}{\gamma -\beta
-2\beta ^{2}+\lambda _{1}\beta ^{2}}\approx 1
\end{eqnarray*}%
such that $\Delta _{1}(x_{1}=0)\approx -1$. It follows that $\Delta
_{1}(x_{1}=1)>c$ and $\Delta _{1}(x_{1}=0)<-c$, such that type $1$ strictly
prefers to play $a_{i}=x_{i}$ for all $x_{i}\in \{0,1\}$. This is consistent
with the postulated strategy.

Finally, note that $p(t=1\mid a,x_{1}=\#)=0$ for both $a=0,1$, hence $\Delta
_{1}(x_{1}=\#)=0$. It is therefore optimal for type $1$ to play $a=0$ when $%
x_{1}=\#$. Since he follows this prescription with probability $%
1-\varepsilon $, this completes the confirmation that type $1$'s behavior is
consistent with $\varepsilon $-equilibrium. By symmetry, the same
calculation holds for type $2$. We have thus constructed an $\varepsilon $%
-equilibrium in which the DM commits an error with probability arbitrarily
close to $\gamma $. Since $c$ can be arbitrarily close to $1$, this
completes the proof.\bigskip 
\end{proof}

I conjecture that the upper bound obtained in this result is tight. Note
that in order to attain it, I used trembles and also required exogenous $x$
variables to take at least three values. Whether these elements in the
construction are indispensable is an open question.

The final result in this sub-section lifts all restrictions on $(p(x,y\mid
t))$ and $P$ and shows that in this case, the gap between equilibrium and
non-equilibrium upper bounds on the DM's welfare loss disappears.\bigskip 

\begin{proposition}
Suppose that all data types are simple and that $P$ is incomplete. Then, for
every $\gamma ,c\in (0,1)$, there exist $\lambda $ and $(p(x,y\mid t))$ for
which there is an equilibrium in which $\Pr (a\neq t)=1$.
\end{proposition}

\begin{proof}
Since $P$ is incomplete, $K\geq 2$. Moreover, there exist two data types, $1$
and $2$, and two exogenous variables, conveniently denoted $x_{1}$ and $%
x_{2} $, such that $1\in C_{1}\setminus C_{2}$ and $2\in C_{2}\setminus
C_{1} $. Let $\lambda _{1}=\lambda _{2}=0.5$. Construct a distribution $p$
over $t,x_{1},x_{2},y$ given by the following table (suppose that $p$ is
constant over the other variables, such that they can be ignored), where $%
\varepsilon >0$ is arbitrarily small:%
\[
\begin{array}{ccccc}
p(t,x_{1},x_{2},y) & t & x_{1} & x_{2} & y \\ 
1-\gamma -\varepsilon & 0 & 1 & 1 & 1 \\ 
\gamma -\varepsilon & 1 & 0 & 0 & 1 \\ 
\varepsilon & 0 & 1 & 0 & 0 \\ 
\varepsilon & 1 & 0 & 1 & 0%
\end{array}%
\]

Suppose data type $i$ plays $a_{i}\equiv x_{i}$. Let us calculate $\Delta
_{1}(x_{1})$ for each $x_{1}$. First,%
\begin{eqnarray*}
p(y &=&1\mid a=1,x_{1}=1)=\frac{1-\gamma -\varepsilon }{1-\gamma
-\varepsilon +\varepsilon \cdot 0.5}\approx 1 \\
p(y &=&1\mid a=0,x_{1}=1)=0
\end{eqnarray*}%
where the second equation holds because the combination of $a=0$ and $%
x_{1}=1 $ occurs only when $x_{2}=0$, in which case $y=0$ with certainty.

Second,%
\begin{eqnarray*}
p(y &=&1\mid a=0,x_{1}=0)=\frac{\gamma -\varepsilon }{\gamma -\varepsilon
+\varepsilon \cdot 0.5} \\
p(y &=&1\mid a=1,x_{1}=0)=0
\end{eqnarray*}%
where the second equation holds because the combination of $a=1$ and $%
x_{1}=0 $ occurs only when $x_{2}=1$, in which case $y=0$ with certainty.

Plugging these terms into the definition of $\Delta _{1}(x_{1})$ yields $%
\Delta _{1}(x_{1}=1)\approx 1$ and $\Delta _{1}(x_{1}=0)\approx -1$. The
calculation for type $2$ is identical due to symmetry. Therefore, for every $%
c<1$, we can set $\varepsilon $ such that each data type $i$ will indeed
prefer to play $a\equiv x_{i}$. Furthermore, for both types $i$, $%
x_{i}=1-t_{i}$ with probability arbitrarily close to one. Therefore, the DM
plays $a=1-t$ with arbitrarily high probability, such that the expected
welfare loss is arbitrarily close to one.\bigskip
\end{proof}

The partial results in this section leave three open problems. First, is $%
\max \{\gamma ,1-\gamma \}$ indeed the tight upper bound on the expected
equilibrium welfare loss when $y$ is purely a function of $t$? Second, does
the upper bound of $\gamma (1-\gamma )$ obtained for complete $P$ in
Proposition \ref{prop bound linear ordering} extend to arbitrary
distributions $p$? Finally, do the results extend to general (non-simple)
data types?

\section{Consequential Actions}

So far, we focused on the extreme case in which the DM's action has a null
objective causal effect on the outcome. This facilitated the definition of
the DM's equilibrium welfare loss due to poor controls. In this section I
extend the analysis to situations in which actions do influence outcome.

Define a variable $z$ that takes values in $0$ and $1$, such that the
objective causal model behind the joint distribution over $t,x,z,a,y$ is
given by the DAG%
\[
\begin{array}{ccc}
(t,x) & \rightarrow & a \\ 
\downarrow &  & \downarrow \\ 
z & \rightarrow & y%
\end{array}%
\]%
That is, $t$ and $x$ are exogenous, as before. The action $a$ is a
consequence of $(t,x)$, via the DM types' strategies. The variable $z$ is
also a consequence of $(t,x)$, independently of $a$ (just as $y$ was in the
baseline model). The outcome $y$ is purely caused by $a$ and $z$, according
to the following conditional probability:%
\[
p(y=1\mid do(a),z)=\beta a+(1-\beta )z 
\]%
where $\beta \in (0,1)$.

This formulation implies that for every type $i$, the perceived outcome of
actions is given by%
\[
\tilde{p}_{i}(y=1\mid x_{C_{i}},do(a))=\beta a+(1-\beta )\tilde{p}%
_{i}(z=1\mid x_{C_{i}},do(a)) 
\]%
where the last term is defined just as in the baseline model:%
\[
\tilde{p}_{i}(z=1\mid x_{C_{i}},do(a))=\sum_{x_{D_{i}}}p(x_{D_{i}}\mid
x_{C_{i}})p(z=1\mid a,x_{D_{i}}) 
\]%
The type's estimated causal effect of $a$ on $z$ given $x$ is%
\[
\Delta _{i}^{z}(x)=\tilde{p}_{i}(z=1\mid x_{C_{i}},do(a=1))-\tilde{p}%
_{i}(z=1\mid x_{C_{i}},do(a=0)) 
\]%
Since $z\perp a\mid (t,x)$, the equilibrium analysis of $\Delta _{i}^{z}(x)$
and how it relates to the DM's strategy is the same as the analysis of $%
\Delta _{i}(x)$ in the baseline model.

It follows that the only thing that needs adjustment is the definition of
the DM's welfare loss. The optimal rational-expectations action maximizes%
\[
\beta a-c\cdot \mathbf{1}[a\neq t] 
\]%
because $a$ has no causal effect on $z$, such that the only effect of $a$ on 
$y$ is via the direct channel parameterized by $\beta $. Therefore, the
expected welfare loss given a joint distribution $p$ is%
\[
\gamma \cdot p(a=0\mid t=1)\cdot (c+\beta )+(1-\gamma )\cdot p(a=1\mid
t=0)\cdot (c-\beta ) 
\]%
Note that in equilibrium, the DM chooses $a=0$ at $t=1$ and $x$ only if $%
c+\beta <-(1-\beta )\Delta _{i}^{z}(x)$. Likewise, the DM chooses $a=1$ at $%
t=0$ and $x$ only if $c-\beta <(1-\beta )\Delta _{i}^{z}(x)$. Consequently,
the upper bounds on the DM's equilibrium welfare loss are the same as in
Sections 3-4, multiplied by $1-\beta $.\bigskip

\noindent \textit{An example: Partying during a pandemic}

\noindent Although the paper emphasized upper bounds on the equilibrium
costs of using bad controls for causal inference, in economic applications
we wish to\ restrict the objective process so that it can capture an
underlying economic reality. I now present a simple example of such an
application.

Suppose that $a=1$ means that the DM chooses to socially distance himself
during a pandemic --- specifically, avoiding parties. The outcome $y=1$
represents good health. Let $x$ represent the DM's age ($x=1$ indicates an
old DM). Let $t$ represent the DM's intrinsic taste for partying --- $t=1$
means that the DM prefers not to go to parties. Let $c<\frac{1}{2}$.

The objective distribution $p$ satisfies: $p(x=1)=\frac{1}{2}$; $p(t=x\mid
x)=q$ for all $x$, where $q\in (\frac{1}{2},1)$; and $p(y=1\mid a,x)=\frac{1%
}{2}(a+1-x)$. This distribution is consistent with the DAG%
\[
\begin{array}{ccc}
t & \leftarrow & x \\ 
\downarrow & \swarrow & \downarrow \\ 
a & \rightarrow & y%
\end{array}%
\]%
That is, $y$ is only caused by $a$ and $x$. When an old DM goes to parties,
his health outcome is bad with certainty; when a young DM avoids parties,
his health outcome is good with certainty; in all other cases, the DM's
health outcome is equally like to be good or bad.

Data type $1$ controls for $x$. This type correctly estimates the causal
effect of switching from $a=0$ to $a=1$ on $y$ to be $\frac{1}{2}$. Since $c<%
\frac{1}{2}$, this DM data type will rationally play $a=1$, independently of 
$t$ and $x$.

Data type $2$ does not control for $x$ (recall that even if it is obviously
natural to assume that the DM knows his age group, the DM may lack
statistics about the age dependence of the correlation between $a$ and $y$,
and therefore cannot use his knowledge of his age). This DM chooses $a$ to
maximize%
\[
p(y=1\mid a)-c\cdot \mathbf{1}[a\neq t]=\frac{1}{2}+\frac{1}{2}a-\frac{1}{2}%
p(x=1\mid a)-c\cdot \mathbf{1}[a\neq t]
\]

Let us analyze equilibria in this example. As we saw, data type $1$'s
strategy is $\sigma _{1}(a=1\mid t)=1$ for all $t,x$. Denote $\sigma
_{2}(a=1\mid t)=\alpha _{t}$ (recall that type $2$ does not condition his
action on $x$). Then,%
\begin{eqnarray*}
p(x &=&1\mid a=1)=\frac{\lambda _{1}+\lambda _{2}[q\alpha _{1}+(1-q)\alpha
_{0}]}{2\lambda _{1}+\lambda _{2}[\alpha _{1}+\alpha _{0}]} \\
p(x &=&1\mid a=0)=\frac{1-q\alpha _{1}-(1-q)\alpha _{0}}{2-\alpha
_{1}-\alpha _{0}}
\end{eqnarray*}

First, let us guess%
\[
p(x=1\mid a=1)-p(x=1\mid a=0)<\frac{1}{2}-c
\]%
Then, $a=1$ is optimal for data type $2$ regardless of $t$. In this case, we
need to consider perturbed strategies to ensure that $p(x=1\mid a=0)$ is
well-defined. Since $\alpha _{0}$ and $\alpha _{1}$ are arbitrarily close to 
$1$, we obtain $p(x=1\mid a=1)\approx \frac{1}{2}$, whereas we can set the
perturbations such that $p(x=1\mid a=0)$ can take any value in $(1-q,q)$. It
follows that it is always possible to sustain the guess in equilibrium, such
that the DM will commit no error.

Second, let us guess%
\[
p(x=1\mid a=1)-p(x=1\mid a=0)>\frac{1}{2}-c
\]%
Then, data type $2$ will play $\alpha _{t}\equiv t$ in equilibrium. Plugging
this into the expressions for $p(x=1\mid a)$, we obtain%
\[
p(x=1\mid a=1)-p(x=1\mid a=0)=\frac{\lambda _{1}+\lambda _{2}q}{2\lambda
_{1}+\lambda _{2}}-(1-q)
\]%
It follows that if%
\[
c>\frac{1}{2}-\frac{2q-1}{1+\lambda _{1}}
\]%
then the guess is consistent. In this case, there is an equilibrium in which
type $2$ follows his taste.

What sustains this equilibrium is the positive correlation between age and
preferences. Young DMs like going to parties more than old DMs, and since
the DM chooses according to his intrinsic taste with some probability ($%
\lambda _{2}$), there is positive correlation between attending parties and
young age. In turn, this soften the negative correlation between $a$ and $y$%
, to an extent that makes it optimal for type $2$ DMs to follow their taste.
The expected welfare loss in this equilibrium is%
\[
\frac{1}{2}\cdot \lambda _{2}\cdot (\frac{1}{2}-c)<\frac{(2q-1)(1-\lambda
_{1})}{2(1+\lambda _{1})}
\]

Note that the welfare loss increases with the fraction of type $2$. There
are two forces behind this observation. First, higher $\lambda _{2}$
obviously means that there are more DMs in the population who are prone to
error. Second, type $1$ DMs do not vary their behavior with $t$ (and hence
with $x$), thus curbing the overall positive correlation between $a$ and $x$
that leads type $2$ DMs to underestimate the causal effect of $a$ on $y$.
The latter effect is a beneficial \textquotedblleft \textit{equilibrium
externality}\textquotedblright\ that the sophisticated DM type exerts on the
naive type: A larger share of the sophisticated type implies that the naive
type commits a smaller error. If public health authorities could somehow
\textquotedblleft educate\textquotedblright\ part of the population to
reason better about causality, this would have a \textquotedblleft
multiplier\textquotedblright\ effect thanks to this equilibrium externality.

There is potentially a third equilibrium in which $\alpha _{1}=1$ and $%
\alpha _{0}\in (0,1)$, such that%
\[
p(x=1\mid a=1)-p(x=1\mid a=0)=\frac{1}{2}-c 
\]%
For brevity, I omit the full characterization of this equilibrium.

\section{Relation to Other Solution Concepts}

The model of behavioral causal inference presented in this paper poses a new
question. However, it can be formulated by adapting existing frameworks of
equilibrium modeling with non-rational expectations.

Jehiel's (2005) concept of analogy-based expectations equilibrium (ABEE)
captures the idea that players' perception of other players' strategies is
coarse. In the present context, we can regard $y$ as the action taken by a
fictitious opponent of the DM after observing the history $%
(a,t,x_{1},...,x_{n})$. In this context, $C_{i}$ defines type $i$'s
information set, whereas $D_{i}$ defines type $i$'s \textquotedblleft
analogy partition\textquotedblright\ (to use Jehiel's terminology). Two
histories belong to the same cell in this partition if they share the same
value of $x_{D_{i}}$. My definition of equilibrium is consistent with
Jehiel's assumption that type $i$ believes that the fictitious player's
strategy is measurable with respect to type $i$'s analogy partition. Thus,
the model in this paper can be formulated as an application of ABEE.

The model can also be cast in the Bayesian-network language of Spiegler
(2016). The objective distribution $p$ in the baseline model (where $a$ has
no causal effect on $y$) is consistent with the following DAG:%
\[
\begin{array}{ccccc}
a & \leftarrow & t & \rightarrow & y \\ 
& \nwarrow & \uparrow & \nearrow &  \\ 
&  & x &  & 
\end{array}%
\]%
Using the DAG language, the distinction between data types in the present
model can be redefined in terms of subjective causal models. Specifically,
type $i$ believes in a causal model that involving the variables on which he
has data, and is given by the following DAG:

\[
\begin{array}{ccc}
x_{D_{i}\setminus C_{i}} & \longrightarrow & y \\ 
\uparrow & \nearrow & \uparrow \\ 
x_{C_{i}} & \longrightarrow & a%
\end{array}%
\]%
According to Spiegler (2016), the subjective belief that this model
generates obeys the Bayesian-network factorization formula%
\[
p(x_{C_{i}})p(x_{D_{i}\setminus C_{i}}\mid x_{C_{i}})p(a\mid
x_{C_{i}})p(y\mid a,x_{C_{i}},x_{D_{i}}) 
\]%
The DM's conditional belief over $y$ as a consequence of $a$ given $%
x_{C_{i}} $ is described by (\ref{subjective belief}). Equilibrium in the
present model is consistent with the notion of personal equilibrium in
Spiegler (2016,2020) when the DM's subjective causal model is random.

The Bayesian-network framework in Spiegler (2016) can be subsumed into the
more general concept of Berk-Nash equilibrium due to Esponda and Pouzo
(2016). According to this concept, the DM best-replies to a conditional
belief (over outcomes given actions and signals), which minimizes a weighted
version of Kullback-Leibler divergence with respect to the objective
conditional distribution. The proper adaptation of this concept to the
present context requires the weights to be given by the DM's ex-ante
equilibrium strategy.

The fact that the present model can be reformulated using existing
frameworks means that fundamentally, it does not present a methodological
innovation. The reason I nevertheless chose to present it in a different
language is twofold. First, this mode of exposition makes it self-contained
and therefore easier to follow for readers who may not know the previous
frameworks. Second, by drawing a connection with the familiar notion of
\textquotedblleft bad controls\textquotedblright , this paper will hopefully
help inspiring new questions about how everyday decision makers perform
causal inference.

\section{Conclusion}

When DMs draw causal inferences from observed correlations, they may commit
errors if they fail to control for an appropriate set of confounding
variables. This paper examined a model of this error, when DM types differ
in their sets of control variables. The main theme of the paper was that
since DMs' causal inferences drive how they condition their actions on their
signals, and since this response in turn shapes the very correlations from
which DMs draw their inferences, equilibrium analysis is required in order
to evaluate the DM's expected decision cost of erroneous causal inference
due to poor controls.

The main general insight that emerged from this analysis was that the upper
bound on this decision cost depends on whether DM types are differentiated
\textquotedblleft vertically\textquotedblright\ or \textquotedblleft
horizontally\textquotedblright . In the former case, types can be partially
ordered according in some sense according to the size of their control
variables. In this case, the equilibrium cost of bad controls is
significantly lower than the non-equilibrium benchmark, and sometimes it is
null. In the former case, types control for different variables, which can
give rise to mutually reinforcing confounding patterns, such that the
maximal equilibrium decision cost is significantly higher than in the former
case, and sometimes coincides with the non-equilibrium benchmark.


\bigskip 

\begin{thebibliography}{99}
\bibitem{} De Barreda, I., G. Levy and R. Razin (2022). Persuasion with
Correlation Neglect: A Full Manipulation Result, American Economic Review:
Insights 4, 123-138.

\bibitem{} Simon Galperti. Persuasion: The Art of Changing Worldviews.
American Economic Review, 109(3):996-1031, 2019.

\bibitem{} Cinelli, C., A. Forney and J. Pearl (2020), A Crash Course in
Good and Bad Controls, Sociological Methods \& Research: 00491241221099552.

\bibitem{} Eliaz, K. , R. Spiegler and H. Thysen (2021), Strategic
Interpretations, Journal of Economic Theory 192, Article 105192.

\bibitem{} Galperti, S. (2019), Persuasion: The Art of Changing Worldviews,
American Economic Review 109, 996-1031.

\bibitem{} Glazer, J. and A. Rubinstein (2012), A Model of Persuasion with
Boundedly Rational Agents, Journal of Political Economy 120, 1057--1082.

\bibitem{} Jacob Glazer and Ariel Rubenstein. Complex Questionnaires.
Econometrica, 82:1529-1541, 2014.

\bibitem{} Jehiel, P. (2005), Analogy-Based Expectation Equilibrium, Journal
of Economic theory 123, 81-104.

\bibitem{} Esponda. I. and D. Pouzo (2016), Berk-Nash Equilibrium: A
Framework for Modeling Agents with Misspecified Models,\ Econometrica 84,
1093-1130.

\bibitem{} Hagenbach, J. and F. Koessler (2020), Cheap Talk with Coarse
Understanding, Games and Economic Behavior 124, 105-121.

\bibitem{} Pearl, J. (2009), Causality: Models, Reasoning and Inference,
Cambridge University Press, Cambridge.

\bibitem{} Sen, A. (1969), Quasi-transitivity, Rational Choice and
Collective Decisions, Review of Economic Studies 36, 381-393.

\bibitem{} Schwartzstein, J. and A. Sunderam (2021), Using Models to
Persuade, American Economic Review 111, 276-323.

\bibitem{} Spiegler, R. (2016), Bayesian Networks and Boundedly Rational
Expectations,\ Quarterly Journal of Economics 131, 1243-1290.

\bibitem{} Spiegler, R. (2020), Behavioral Implications of Causal
Misperceptions, Annual Review of Economics 12, 81-106.

\bibitem{} Spiegler, R. (2022), On the Behavioral Consequences of Reverse
Causality, European Economic Review 149: 104258.
\end{thebibliography}
\end{document}